\documentclass[10pt,a4paper]{article}
\usepackage[latin1]{inputenc}
\usepackage{algorithm}
\usepackage{placeins}
\usepackage{algpseudocode}
\usepackage{amsmath}
\usepackage{amsthm}
\usepackage{amsfonts}
\usepackage{mathdots}
\usepackage{amssymb}
\usepackage{enumerate}
\usepackage{xcolor}
\usepackage[hidelinks]{hyperref}
\hypersetup{colorlinks,linkcolor={red!50!black},citecolor={blue!50!black},urlcolor={green!50!black}}

\allowdisplaybreaks
\theoremstyle{plain}
\newtheorem{thm}{Theorem}
\newtheorem{lem}{Lemma}
\newtheorem{pro}{Proposition}
\newtheorem{cor}{Corollary}
\theoremstyle{definition}
\newtheorem{defn}{Definition}
\theoremstyle{remark}
\newtheorem{rem}{Remark}
\newtheorem{exa}{Example}

\author{Tovohery Hajatiana Randrianarisoa\footnote{The author is supported by SNF grant no. 169510}}
\title{Coding Theory using Linear Complexity of \\ Finite Sequences}

\renewcommand{\a}{\alpha}
\newcommand{\N}{\mathbb{N}}
\newcommand{\F}{\mathbb{F}_q}
\newcommand{\e}{\mathbf{e}}
\newcommand{\ai}{(a_i)}
\newcommand{\bi}{(b_i)}
\newcommand{\ci}{(c_i)}
\renewcommand{\L}{\mathfrak{L}}
\renewcommand{\O}{\mathcal{O}}
\renewcommand{\d}{\mathbf{d}}
\newcommand{\C}{\mathcal{C}}
\renewcommand{\a}{\alpha}

\newcommand{\A}{\mathbf{A}}
\newcommand{\0}{\mathbf{0}}
\newcommand{\U}{\mathbf{U}}
\newcommand{\V}{\mathbf{V}}
\newcommand{\X}{\mathbf{X}}
\newcommand{\Y}{\mathbf{Y}}
\renewcommand{\l}{\lambda}
\newcommand{\x}{\mathbf{x}}
\newcommand{\y}{\mathbf{y}}

\begin{document}
\maketitle

\begin{abstract}
We define a metric on $\F^n$ using the linear complexity of finite sequences. We will then develop a coding theory for this metric. We will give a Singleton-like bound and we will give constructions of subspaces of $\F^n$ achieving this bound. We will compute the size of balls with respect to this metric. In other words we will count how many finite sequences have linear complexity bounded by some integer $r$. The paper is motivated in part by the desire to design new code based cryptographic systems.
\end{abstract}

\section{Motivation}\label{sec:1}

As we will explain in this section, the notion of weight of vectors are closely related to notion of linear complexity of the sequence. This motivates us to study the linear complexity of sequences as a new metric. For us to see this relation, let us first recall some notion from linear coding theory using the Hamming metric. 

In most part of this work, unless otherwise specified, we will always work with a finite field $\F$ of size $q$. 

\begin{defn}\label{defn:1}
Let $\x = (x_1,\cdots,x_n)\in \F^n$. We define the weight $\omega(\x)$ of $\x$ to be the number of non-zero entries of $\x$. If $\x$ and $\y$ are two elements of $\F^n$, then we define the distance between $\mathbf{a}$ and $\mathbf{b}$ as $\d(\x,\y)=\omega(\x-\y)$.
\end{defn}

\begin{defn}\label{defn:2}
A linear code $\C$ of length $n$ over $\F$ is a subspace of $\F^n$ paired with the distance $\d$ as in the previous definition. The minimum distance of a linear code is the smallest value of $d(\x,\y)$ for any two distinct codewords of $\C$.
\end{defn}
The most important parameters for a linear code $\C$ are the size of the base field, the length, the dimension and its minimum distance. One has to optimize the choice of these parameters for applications. For example, one wants to construct codes with large dimension and large minimum distance at the same time and the base field should preferably be as small as possible (binary field for example). This is not an easy task as the minimum distance behaves in opposite to the dimension for example. This is explained by the following Singleton bound.
\begin{thm}[Singleton bound]\label{thm:1}
Let $\C$ be a linear code of length $n$ and dimension $k$ as subspace of $\F^n$. Suppose that $d$ is the minimum distance of $\C$. Then,
\[
d\leq n-k+1.
\]
\end{thm}

Due to this, we want to have codes which maximize both the dimension and the minimum distance of the code. Thus we want to have codes for which the inequality in the above definition is an equality. Such codes are defined as follows.

\begin{defn}\label{defn:3}
A linear code $\C$ which attains the Singleton bound i.e if $\C$ is of dimension $k$, $d$ is its minimum distance and $d=n-k+1$, is called a maximum distance separable (MDS) code. 
\end{defn}
Maximum distance separable codes exist. One easy construction is given by the following. Let $n=q-1$ and let $\a = (\a_1,\cdots,\a_n)$ be a vector where its elements are the non-zero elements of $\F$. We define the evaluation map as
\begin{align*}
ev_\a: \F[x] & \rightarrow\F^n \\
f(x) & \mapsto(f(\a_1),\cdots,f(\a_n))
\end{align*}

Let $\F[x]_{< k}$ be the vector space of all polynomials of degree at most $k-1$. Then the image $\C = ev_\a\left(\F[x]_{< k}\right)$ is an MDS code. This comes from the fact that a polynomial of degree at most $k-1$ can have at most $k-1$ roots. The code we described is called Reed-Solomon code. 

It is this relation between the property of the roots of polynomial which is interesting for us. Let us see the following theorem of K\"onig-Rados. For a proof of this theorem, one can have a look at Chapter 6 of \cite{Lid96}.
\begin{thm}[K\"onig-Rados]\label{thm:2}
Let $f(x)=a_0 +a_1 x + \cdots + a_{q-2}x^{q-2}$ be a polynomial over $\F$. Define the following matrix 
\[
\A = 
\begin{pmatrix}
a_0 & a_1 & \hdots & a_{q-2} \\
a_1 & \iddots & \iddots & a_{0} \\
\vdots & \iddots & \iddots & \vdots \\
a_{q-2} &  a_{0} & \hdots & a_{q-3}
\end{pmatrix}.
\]
Suppose that the rank of $\A$ is equal to $r$. Then the number of roots of $f(x)$ in $\F^*$ is given by $q-1-r$.
\end{thm}

The matrix $\A$ in the above theorem is a circulant matrix. It is easy to see that if its rank is equal to $r$, then the first $r$ rows of $\A$ are linearly independent and the other rows are linear combination of them. Furthermore, this tells us that the coefficients of $f(x)$ satisfy the following property.
\[
a_{i+r} = \sum_{j=0}^{r-1} c_j a_{i+j},\quad \forall i\in \N.
\]
Note that the coefficients $a_0,\cdots,a_{q-2}$ satisfy a recurrence relation of order $r$. Using the definitions which we will see in Section \ref{sec:2}, we say that the coefficients of the polynomials $f(x)$ can be generated by a linear-feedback shift register (LFSR) of length $r$ and this is the minimum possible for $r$. We say that $(a_0,\cdots,a_{q-2})$ has linear complexity $r$. Moreover, our sequence gives a periodic sequence with period $q-1$. To summarize, we have the following theorem, which is a direct consequence of the theorem of K\"onig-Rados.

\begin{thm}\label{thm:3}
Let $f(x) = a_0 +a_1 x + \cdots + a_{q-2}x^{q-2}$ be a polynomial over $\F$. If $f(x)$ has $q-1-r$ roots, then $(a_0,\cdots,a_{q-2})$ has linear complexity $r$ and the evaluation  $(f(\a_1),\cdots,f(\a_n))$ has weight $r$.
\end{thm}

Through Theorem \ref{thm:3}, we can relate the linear complexity of a periodic sequence with the weight of a vector. However, we have only periodic sequences. This raises the following question: What happens if we study any type of sequence i.e. we don't require the LFSR to be a periodic sequence with fixed period. We will answer this question in the next sections. First, in Section \ref{sec:2}, we will introduce the notion of linear-feedback shift register. In Section \ref{sec:3}, we will give a coding theory for finite sequences. We will use Section \ref{sec:4} for a separate study on the number of finite sequences which can be generated by an LFSR of given length. Finally, we will conclude with Section \ref{sec:5} and give some future work.

\section{Linear-feedback shift register}\label{sec:2}
Let $\F$ be a finite field with $q$ elements.

\begin{defn}\label{defn:4}
Left $\F$ be a field. A linear feedback shift register (LFSR) of order $l$ over $\F$ is an infinite
sequence $\ai$ over $\F$ such that, there are fixed $c_j\in \F$ 
with,
\[
a_{i+l} = \sum_{j=0}^{l-1} c_j a_{i+j},\quad \forall i\in \N.
\]
The feedback polynomial associated to $\ai$ is
\[
f(z) = z^l - \sum_{j=0}^{l-1} c_j z^j.
\]
\end{defn}

\begin{defn}\label{defn:5}
Let $\ai$ be a LFSR over $\F$. The generating function $A(z)$ associated to 
$\ai$ is the formal power series
\[
A(z) = \sum_{i=0}^\infty a_i z^i.
\]
\end{defn}

One can show (Chapter 8 \cite{Lid96}) that for some polynomial $g(z)$ of degree $l-1$ at most, we have
\[
A(z) = \frac{g(z)}{f^*(z)},
\]
where $f^*$ is the reciprocal polynomial given by
\[
f^*(z) = z^l f\left(\frac{1}{z}\right).
\]

\begin{defn}\label{defn:6}
Given a non-zero finite sequence $\ai = (a_0,\cdots,a_{n-1})\in \F^n$, the linear complexity  $\L(a_i)$ of the sequence is the smallest $l$ such that
\[
a_{i+l} = \sum_{j=0}^{l-1} c_j a_{i+j},\quad \forall i, \; 0\leq i\leq n-1,
\]
for some fixed $c_j\in \F$. 

For a zero sequence, we set the linear complexity to be equal to zero.
\end{defn}

Given a finite sequence, we can compute the shortest LFSR 
producing this sequence. This can be done using the Berlekamp-Massey algorithm in 
$\O(n^2)$ field operations in $\F$ (Chapter 8 of \cite{Lid96}). Furthermore if the linear complexity is $n/2$, then $n$ successive terms of the sequence are enough to uniquely find the shortest shift register. We present the algorithm in Algorithm \ref{algo:1}. On input, we have a sequence $s_0,\cdots,s_{n-1}$ of length $n$. On output, the algorithm generates the order and the feedback polynomial $f(z)$ of the shortest LFSR generating $s_0,\cdots,s_{n-1}$.

\begin{algorithm}[ht!]
\caption{Berlekamp-Massey}\label{algo:1}
\begin{algorithmic}[1]
\Procedure{BERLEKAMP-MASSEY}{$s_0,\cdots,s_{n-1}$}
\State $f(z)\gets 1$, $A(z)\gets 1$, 
\State $L \gets 0$, $m=-1$, $e \gets 1$
\For{$i$ from $0$ to $n-1$}
\State $d \gets s_i + \sum_{j = 1}^{L}f_j s_{i-j}$
\If{$d\neq 0$}
\State $B(z)\gets f(z)$
\State $f(z)\gets f(z) -  (d/e)A(z)z^{i-m}$
\If{$2L\leq i$}
\State $L\gets i+1-L$
\State $m \gets i$
\State $A(z)\gets B(z)$
\State $e\gets d$
\EndIf
\EndIf
\EndFor
\State \Return $L$ and $f(z)$
\EndProcedure
\end{algorithmic}
\end{algorithm}

\begin{pro}\label{pro:1}
Let $\ai = (a_0,\cdots,a_{n-1})$ be a finite sequence of length $n$. Then $\L\ai \leq n$. Furthermore the only sequences attaining the bound upper bound $n$ are of the form $(0,\cdots,0,a)$, with $a\in \F^*$.
\end{pro}
\begin{proof}
We can just use a LFSR with $\ai$ as initial state so that the maximum linear complexity is at most $n$. It is obvious that $(0,\cdots,0,a)$ has linear complexity $n$. Finally, if $\ai = (a_0,\cdots,a_{n-1})$ is such that $a_j \neq 0$ for some $j$ with $0\leq j\leq n-2$. By taking $c_i = 0$ except when $i = j$, where $c_j = a_{n-1}/a_j$, we prove that $a_{n-1} = \sum_{j=0}^{n-2} c_j a_{j}$ so that the linear complexity is at least $n-1$.
\end{proof}

The key property of the linear complexity of sequences which will be used later is the following.
\begin{thm}\label{thm:4}
Let $\ai$ and $\bi$ be two finite sequences. If $\ci = \ai + \bi$, then
\[
\L\ci\leq \L\ai + \L\bi.
\]
\end{thm}
\begin{proof}
Suppose that the generating function of the LFSR generating $\ai$ and $\bi$ are respectively
\[
\frac{g_a(z)}{f_a^*(z)},\quad \text{ and } \frac{g_b(z)}{f_b^*(z)}.
\]
Then the generating function of the LFSR generating $\ci$ is
\[
\frac{g_a(z)f_b^*(z)+g_b(z)f_a^*(z)}{f_a^*(z)f_b^*(z)}.
\]
And therefore, $\ci$ can be generated by a LFSR with feedback 
polynomial $f_a(z)f_b(z)$. Therefore the linear complexity is at most $\L\ai +\L\bi$.
\end{proof}

\section{A coding theory for finite sequences using the linear complexity}\label{sec:3}
Let $\F$ be a finite field and let $n$ be a positive integer. We will consider sets of length $n$.

\begin{defn}\label{defn:7}
Let $\ai = (a_0,\cdots,a_{n-1})\in \F^n$ and $\bi = (b_0,\cdots,b_{n-1})\in\F^n$ be two finite sequences of $n$ elements of $\F$ each. Then we define a distance on $\F^n$ by the following,
\[
\d(\ai,\bi) = \L(\ai - \bi),
\]
where $\L(0) = 0$.
\end{defn}

This map defines indeed a distance:
\begin{enumerate}[(i)]
\item By definition $\d(\ai,\bi) = 0 \Leftrightarrow \ai = \bi$.
\item By definition of $\L$, $\L\ai \geq 0$.
\item $\d(\ai,\bi) = \d(\bi,\ai)$.
\item For the triangular inequality,
\begin{align*}
\d(\ai,\bi) &= \L(\ai - \bi) \\
& = \L(\ai -\ci +\ci - \bi) \\
& \leq \L(\ai -\ci) + \L(\ci - \bi), \text{ by Theorem \ref{thm:4}}\\
& = \d(\ai,\ci) + \d(\ci,\bi).
\end{align*}
\end{enumerate}

Like in coding theory, we can define a subset of $\F^n$ and define the metric $\d$ on this set. We will derive basic coding results for this context.

\begin{defn}\label{defn:8}
Let $S$ be a subset of $\F^n$. The minimum distance $d$ of $S$ is the minimum 
of $\d(\ai,\bi)$ for distinct $\ai,\bi\in S$. We will describe the parameters 
of $S$ as $[n,\sharp S,d]$. In case $S$ is a $k$-dimensional subspace of $\F^n$, 
then, by additivity, $d$ is the minimum linear complexity of the non-zero sequences in $S$ and we 
will write $[n,k,d]$.
\end{defn}

For the next steps we want to have a look at the bounds on a  $[n,\sharp S,d]$-subset of $\F^n$.
\begin{thm}[Singleton bound]\label{thm:5}
Let $\F$ be a finite field of size $q$. Let $S\subset \F^n$ be a set of finite sequence over $\F$ of length $n$, with minimum distance $d$. Then $\sharp S\leq q^{n-d+1}$.
\end{thm}
\begin{proof}
Let us define the following linear map $P$ as
\begin{align*}
P:\F^n&\rightarrow \F^{n-d+1} \\
(a_0,\cdots,a_{n-1})&\rightarrow \begin{pmatrix}
1 & \hdots & 1
\end{pmatrix}\begin{pmatrix}
a_0 & \hdots & a_{n-d} \\
\vdots & \ddots & \vdots \\
a_{d-1}& \hdots & a_{n-1}
\end{pmatrix}
\end{align*}
The restriction of this map must be injective on $S$. Otherwise if two sequences
$\ai$ and $\bi$ are mapped to the same image, then $\ai-\bi$ is mapped to zero.
But this would imply that $\L\left(\ai-\bi\right)\leq d-1$. This is in
contradiction with the minimum distance of $S$. 
By the injectivity, we must have that $\sharp S\leq \sharp(\F^{n-d+1})$.
\end{proof}

Note that in this proof, instead of using $\begin{pmatrix}
1 & \hdots & 1
\end{pmatrix}$, we can use any vector with $1$ as last entry. These operations are equivalent to the puncturing operation on codes. Namely using $\begin{pmatrix}
0 & \hdots & 0 & 1
\end{pmatrix}$ is just puncturing at the first $d-1$ positions.
\begin{rem}\label{rem:1}
In case $S$ is linear of dimension $k$ over $\F$, then $k\leq n-d+1$.
\end{rem}

\begin{defn}[Optimal set of sequences - OSS]\label{defn:9}
We call a set of sequences $S$ optimal if the minimum distance of $S$ reaches the bound of the previous theorem i.e. if $S$ has elements of length $n$ and minimum distance $d$ and $\sharp S = q^{n-d+1}$.
\end{defn}

\begin{exa}\label{exa:1}
Let $S$ be the set of sequences of length $n$ over a finite field $\F$ defined by 
\[
S = \lbrace (0,\cdots,0,a_1,\cdots,a_k): a_i\in \F \rbrace.
\]
Then, $S$ is an optimal set of linear sequences of dimension $k$. That is because the sequences cannot be generated by a LFSR of length smaller than $n-k+1$ except when it is the zero sequence.
\end{exa}

The nice property of using the set of sequences with the linear complexity as metric is that, in opposite to maximum distance separable codes, we can have optimal set of sequences for any parameters. We can make the construction, even for the binary field.

\subsection*{Decoding of OSS}
The decoding of OSS given in Example \ref{exa:1} is straightforward. First let us look at the unique decoding property.

\begin{pro}\label{pro:2}
Suppose that $S$ is an $[n,M,d]$ set of sequences. Suppose that $\y\in \F$ is equal to $\x+\e$, where $\x\in S$ and $\L(\e)<\frac{d}{2}$. Then, the decomposition $\x+\e$ is unique.
\end{pro}
\begin{proof}
If $ \y = \x_1+\e_1 = \y_2+\e_2$, then $\x_1-\x_2 = \e_2-\e_1$. Therefore $d(x_1,x_2) = \L(\e_2-\e_1)$. By Theorem \ref{thm:4}, $d(x_1,x_2)\leq \L(\e_2)+\L(\e_1)<d$. This is in contradiction with the minimum distance of $S$.
\end{proof}

Let $S$, of dimension $k$, be the OSS in Example \ref{exa:1}. Suppose that we know $\y = \x+\e$ with $\x\in S$ and $\L(\e)<\frac{n-k+1}{2}$. 
By Proposition \ref{pro:2}, we know that $\e$ is unique. Since the $n-k$ first entries of $\x$ are equal to zero. Then we know the first $n-k$ entries of $\e$. Now, since $\L(\e)<\frac{n-k+1}{2}$, then we can uniquely recover the LFSR generating $\e$ by using the Berlekamp-Massey algorithm on the first $n-k$ entries of $e$. We are therefore able to produce the whole $\e$ and then we compute $\x = \y-\e$.

\begin{rem}\label{rem:2}
We can modify the above decoding algorithm to get a decoding algorithm for the Reed-Solomon code in Section \ref{sec:1}. The extra step is just that we need to interpolate a received codewords first to get a polynomial $f(x)$ of degree $q-2$ at most. After this we apply the decoding algorithm for the OSS we gave above on the coefficients of this polynomial $f(x)$. Notice that the Berlekamp-Massey in this case is applied to the last coefficients of the polynomial $f(x)$.

\end{rem}

\section{Number of finite sequences generated by a LFSR with fixed length}\label{sec:4}

LFSR already has applications in cryptography. For instance, it is used when one wants to generate random keys. As we have seen, one can compute the linear complexity of a sequence using the Berlekamp-Massey algorithm. Thus, if a sequence has small linear complexity, one can easily find a LFSR generating this sequence. Due to this fact, we usually want to have sequences with large linear complexity. Therefore, one important question is to know how many finite sequences  have large linear complexity. Another motivation for this section is also that knowing the number of sequences with a given linear complexity is important for the security aspect of a code-based cryptosystem using linear complexity as metric.

\begin{lem}\label{lem:1}
Let $\ai$ be an infinite sequence. If $\ai$ can be generated by a LFSR of length $n$, then $\ai$ can be generated by a LFSR of length $i$, for any $i\geq n$.
\end{lem}
\begin{proof}
For a proof of this, if $c_1,\cdots,c_{n-1}$ are the coefficients of the LFSR of length $n$, then $0,\cdots,0,c_1,\cdots,c_{n-1}$ are the larger LFSRs.
\end{proof}

By Lemma \ref{lem:1}, we can just study the number of sequences which can be generated by a LFSRs of length $n$ to know the number of sequences which has linear complexity smaller or equal to $n$. Studying sequences which can be generated by a LFSRs of length $n$ can be in turn translated to studying some matrix $\A$ of the form
\[
\A = \begin{pmatrix}
a_0 & a_1 & \hdots & \hdots & a_{n-r-1} \\
a_1 & \iddots & \iddots & a_{n-r-1} & a_{n-r} \\
\vdots & \iddots & \iddots & \iddots & \vdots \\
a_r &  a_{r+1} & \hdots & \hdots & a_{n-1}
\end{pmatrix}
\]
We just need the condition that the last row is a linear combination of the previous rows. The matrices with the form of $\A$ are called are called Hankel matrices when they are square matrices. In \cite{Day60}, Daykin called the general rectangular matrices {\em persymmetric matrices}. To go further with our counting, we will need the following reduction method as used by Daykin in \cite{Day60}.

Fix and integer $u$ such that $0\leq u< \min(r,n-r-1)$. We define the following set
\[
\mathcal{A}_u = \lbrace(0,\cdots,0,a_u,\cdots,a_{n-1}): a_u\neq 0\rbrace.
\]
Then for $\ai\in \mathcal{A}_u$, we recursively define $\theta_i$, $i=0,\cdots, n-u-1$ by
\[
\begin{cases}
a_u\theta_0 = 1 \\
\sum_{l=0}^i  a_{u+l}\theta_{i-l} = 0
\end{cases}.
\]

Now define the following matrices
 
\[
\U = \begin{pmatrix}
\theta_0 & 0 & 0 & \hdots & 0\\
\theta_1 & \theta_0 & 0 & \ddots & \vdots \\
\vdots & \ddots & \ddots & \ddots & 0 \\
\vdots & \ddots & \ddots & \ddots & 0 \\
\theta_{r} & \hdots & \hdots & \theta_1 & \theta_0
\end{pmatrix},\quad 
\V = \begin{pmatrix}
\theta_0 & \theta_1 & \hdots & \hdots & \theta_{n-r-1}\\
0 & \theta_0 & \theta_1 & \ddots & \vdots\\
0 & \ddots & \ddots & \ddots & \vdots \\
\vdots & \ddots & \ddots & \ddots & \theta_1 \\
0 & \hdots & 0 & 0 & \theta_0
\end{pmatrix}
\]
\[
\X = \begin{pmatrix}
0 & \hdots & 0 & 0 & \theta_0 \\
\vdots & \iddots & 0 & \theta_0 & \theta_1 \\
0 & \iddots & \iddots & \iddots & \vdots \\
0 & \theta_0 & \iddots & \iddots & \vdots \\ 
\theta_0 & \theta_1 & \hdots & \hdots & \theta_{u}\\ 
\end{pmatrix},\quad 
\Y = \begin{pmatrix}
\theta_{u+2}  & \theta_{u+3} & \hdots & \hdots & \theta_{n-r} \\
\theta_{u+3} & \iddots & \iddots & \theta_{n-r-1} & \theta_{n-r+1} \\
\vdots & \iddots & \iddots & \iddots & \vdots \\
\vdots & \theta_{r+1} & \iddots & \iddots & \theta_{n-u-2} \\
\theta_{r+1} & \theta_{r+2} & \hdots & \theta_{n-u-2} & \theta_{n-u-1}
\end{pmatrix}
\]
\begin{lem}\label{lem:2}
For a fixed $u$ with $0\leq u< \min(r,n-r-1)$, there is a bijection between $\mathcal{A}_u$ and the set $\lbrace (\theta_0,\cdots,\theta_{n-u-1}): \theta_0\in \F^*, \theta_i\in \F , 1\leq i\leq n-u-1\rbrace$ given by 
\[
\U\A\V = \begin{pmatrix}
\X & \0 \\
\0 & -\Y
\end{pmatrix}
\]
\end{lem}
\begin{proof}
First let us show that
\[
\U\A\V = \begin{pmatrix}
\X & \0 \\
\0 & -\Y
\end{pmatrix},
\]
First, we know that $\A_{i,k} = a_{i+k}$ for $0\leq i\leq r$ and $0\leq k\leq n-r-1$. For the matrix $\V$, $\V_{j,k} = 0$ if $j<k$ and $\V_{j,k} = \theta_{k-j}$ if $k\leq j$. And for the matrix $\U$, $\U_{i,k} = 0$ if $k>i$ and $\U_{i,k} = \theta_{i-k}$ if $k\leq i$. 
Thus 
\[
(\U\A\V)_{i,j} = \sum_{k=0}^i\theta_{i-k}\left[\sum_{l=0}^j a_{k+l}\theta_{j-l}\right], \quad 0\leq i\leq r, 0\leq j\leq n-r-1.
\]

We are now going to look at three different cases:
\begin{itemize}
\item  Suppose that $i\leq u$, Since $a_0 = a_1 = \cdots =a_{u-1} = 0$, then
\[
(\U\A\V)_{i,j} = \sum_{k=0}^i\theta_{i-k}\left[\sum_{l=u-k}^j a_{k+l}\theta_{j-l}\right], \quad 0\leq i\leq r, 0\leq j\leq n-r-1.
\]
After a change of variable
\[
(\U\A\V)_{i,j} = \sum_{k=0}^i\theta_{i-k}\left[\sum_{l=0}^{j-u+k} a_{u+l}\theta_{j-u+k-l}\right], \quad 0\leq i\leq r, 0\leq j\leq n-r-1.
\]
By the recurrence relation on the $\theta_i$'s, we know that
\[
\sum_{l=0}^{j-u+k} a_{u+l}\theta_{j-u+k-l} =
\begin{cases}
 1, & \text{ if }  j-u+k = 0\\
0, & \text{ otherwise}
\end{cases}
\]
And thus
\[
(\U\A\V)_{i,j} =
\begin{cases}
\theta_{i+j-u} , & \text{ if } 0\leq u-j\leq i \\
0, & \text{ otherwise}
\end{cases}
\]
\item Now, suppose that $i>u$ and $j\leq u$. Since, $j\leq u$, then we can use the expression
\[
(\U\A\V)_{i,j} = \sum_{l=0}^j \theta_{j-l}\left[\sum_{k=0}^i a_{k+l}\theta_{i-k}\right], \quad 0\leq i\leq r, 0\leq j\leq n-r-1.
\]
We use the same transformation as before to get
\[
(\U\A\V)_{i,j} = \sum_{l=0}^j\theta_{j-l}\left[\sum_{k=0}^{i-u+l} a_{u+k}\theta_{i-u+l-k}\right], \quad 0\leq i\leq r, 0\leq j\leq n-r-1.
\]
and
\[
\sum_{k=0}^{i-u+l} a_{u+k}\theta_{i-u+l-k} =
\begin{cases}
 1, & \text{ if }  i-u+l = 0\\
0, & \text{ otherwise}
\end{cases}
\]
Since $u < i$, then the first case is never possible, therefore $(\U\A\V)_{i,j}$ is always zero.
\item Finally, suppose that If $i>u$ and $j> u$. We have 
\begin{align*}
(\U\A\V)_{i,j} & = \sum_{l=0}^j \theta_{j-l}\left[\sum_{k=0}^i a_{k+l}\theta_{i-k}\right] \\
& = (\U\A\V)_{i,u} + \sum_{l=u+1}^j \theta_{j-l}\left[\sum_{k=0}^i  a_{k+l}\theta_{i-k}\right] \\
& = \sum_{l=u+1}^j \theta_{j-l}\left[\sum_{k=u-l}^i  a_{k+l}\theta_{i-k} -\sum_{k=u-l}^{-1}  a_{k+l}\theta_{i-k} \right]
\end{align*}
The last equality contains the subtraction because by starting $k$ with $u-l$, we have some negative value for the index $k$, so we have to remove them. 
Finally, we have
\begin{align*}
(\U\A\V)_{i,j} &= -\sum_{l=u+1}^j \theta_{j-l}\sum_{k=u-l}^{-1}  a_{k+l}\theta_{i-k}\\
& = -\sum_{k=u-j}^{-1} \theta_{i-k} \sum_{l=u-k}^j a_{k+l}\theta_{j-l} \\
& = -\sum_{k=u-j}^{-1} \theta_{i-k} \sum_{l=0}^{j-u+k} a_{u+l}\theta_{j-u+k-l}
\end{align*}
By the recurrence relation on the $\theta_i$'s, we have
\[
(\U\A\V)_{i,j} = -\theta_{i+j-u}.
\]
\end{itemize}

For the bijection, suppose that $\A$ and $\A'$ both give the same $\theta_i$'s, then $\U\A\V = \U\A'\V$, but since $\U$ and $\V$ are invertible, then $\A = \A'$. We have an injection between two sets of the same size, therefore it is a bijection. 
\end{proof}

\begin{lem}\label{lem:3}
Suppose that $\A$ is the matrix corresponding to the sequence $\ai\in \mathcal{A}_u$, and it corresponds to the $\U,\V,\X,\Y$, then the last row of $\A$ is a linear combination of its other rows if and only if the last row of $\Y$ is a linear combination of its other rows.
\end{lem}
\begin{proof}
We know that 
\[
\A\V =  \U^{-1}\begin{pmatrix}
\X & \0 \\
\0 & -\Y
\end{pmatrix}
\]
Thus if $(\l_0,\cdots,\l_{r-1},1)\V = \0$, then 
\[
(\l_0,\cdots,\l_{r-1},1)\U^{-1}\begin{pmatrix}
\X & \0 \\
\0 & -\Y
\end{pmatrix} = \0
\]
Therefore, there is some non-zero $\mu_r$ with,
\[
(\mu_0,\cdots,\mu_{r-1},\mu_{r})\begin{pmatrix}
\X & \0 \\
\0 & -\Y
\end{pmatrix} = \0.
\]
Since $\mu_r\neq 0$, then the last row of $\Y$ is a linear combination of its previous row. The converse can be proven by going backward.
\end{proof}

\begin{defn}\label{defn:10}
Define $B(n,r,u)$ to be the set of non-zero sequences $\ai$ of length $n$ with linear complexity $r$ at most such that $u$ is the smallest index $i$ such that $a_i$ is non-zero. We also define $B(n,r)$ to be the set of all sequences $\ai$ of length $n$ with linear complexity at most $r$. Therefore $B(n,r) = \left(\cup_{u=0}^{n-1} B(n,r,u)\right)\cup \lbrace \0\rbrace$. We set $b(n,r,u) = \sharp B(n,r,u)$ and $b(n,r) = \sharp B(n,r)$.
\end{defn}

Now, suppose that $r+1\leq n-r$. If $u<r$, then we use the above method for reduction. Otherwise if $u\geq r$, then the first $r$ elements of $\ai$ are $0$ an therefore we can only get the zero sequence. 

Next, if $n-r\leq r$, then again we use the above reduction method for $u<n-r-1$. If $r>u\geq n-r-1$, then for any choice of the remaining coefficients $a_{u+1},\cdots, a_{n-1}$, it is always possible to generate it using a LFSR of length $r$. If $u\geq r$, then there is no LFSR of order at most $r$ which can generate the sequence.

These, together with Lemmas \ref{lem:1} and \ref{lem:3} allow us to get the next theorem.

\pagebreak
\begin{thm}\label{thm:6}\ 
\begin{enumerate}[(i)]
\item\label{thm:6-i} If $r+1\leq n-r$ and $u\geq r$, then $b(n,r,u) = 0$.
\item\label{thm:6-ii} If $n-r\leq r$ and $r>u\geq n-r-1$, then $b(n,r,u) = q^{n-u-1}(q-1)$.
\item\label{thm:6-iii}If $n-r\leq r$ and $u\geq r$, then $b(n,r,u) = 0$.
\item\label{thm:6-iv} If $r+1\leq n-r$ and $u<r$, or $n-r\leq r$ and $u<n-r-1$ then $b(n,r,u) = 
q^{u+1}(q-1)b(n-2u-2,r-u-1)$.
\end{enumerate}
\end{thm}
\begin{proof}
Lemma \ref{lem:1} tells us all elements of $B(n,r)$ can be generated by a LFSR of order $r$. Hence, we study only matrices in the form of $\A$. For \eqref{thm:6-i}, one can just write down the matrix $\A$ and see that it has a triangular shape where the last row is never a linear combination of the previous row. For \eqref{thm:6-ii}, we again look at the form of the matrix $\A$, we will see that some first non-zero rows of $\A$ make an invertible matrix and thus the last row is always a linear combination of the rows of that invertible matrix whatever the choice of the coefficients we choose after $a_u$. For \eqref{thm:6-iii}, looking at the form of the matrix will also give the result. For \eqref{thm:6-iv}, we use the bijection in Lemma \ref{lem:2} and Lemma \ref{lem:3}.
\end{proof}

Summing all the possibilities in Theorem \ref{thm:6}, we get the following corollaries.

\begin{cor}\label{cor:1}
Given two integers $r\leq n$, the number of finite sequence of length $n$ with linear complexity at most $r$ is equal to 
\begin{enumerate}[(i)]
\item If $r=0$, $b(n,0) = 1$.
\item If $r+1\leq n-r$,
\[
b(n,r) = 1+\sum_{u=0}^{r-1} q^{u+1}(q-1)b(n-2u-2,r-u-1).
\]
\item If $n-r\leq r$,
\[
b(n,r) = 1+\sum_{u=0}^{n-r-2} q^{u+1}(q-1)b(n-2u-2,r-u-1)+\sum_{u=n-r-1}^{r-1}(q-1)q^{n-u-1}.
\]
\end{enumerate}
\end{cor}

\begin{cor}\label{cor:2}
Given two integers $r\leq n$, the number of finite sequences of length $n$ with linear complexity at most $r$ is equal to $1$ if $r=0$ and if $0<r\leq n$,
\[
b(n,r) = 1-q+ q^2 b(n-2,r-1).
\]
\end{cor}
\begin{proof}
Suppose that $r+1\leq n-r$. Then,
\begin{align*}
b(n,r) &= 1+\sum_{u=0}^{r-1} q^{u+1}(q-1)b(n-2u-2,r-u-1)\\
&= 1 + q(q-1) b(n-2,r-1) + \sum_{u=1}^{r-1} q^{u+1}(q-1)b(n-2u-2,r-u-1)\\
&= 1 + q(q-1) b(n-2,r-1) + q \sum_{u=0}^{r-2} q^{u+1}(q-1)b(n-2u-4,r-u-2)\\
&= 1 + q(q-1) b(n-2,r-1) + q b(n-2,r-1) -q\\
&= 1-q+ q^2 b(n-2,r-1)
\end{align*}

Now suppose that $n-r\leq r$. Then,

\begin{align*}
b(n,r) &= 1 +\sum_{u=n-r-1}^{r-1}(q-1)q^{n-u-1} \\
&\quad + \sum_{u=0}^{n-r-2} q^{u+1}(q-1)b(n-2u-2,r-u-1)\\
&= 1 + \sum_{u=n-r-1}^{r-1}(q-1)q^{n-u-1} + q(q-1)b(n-2,r-1) \\
& \quad + q\sum_{u=1}^{n-r-2} q^{u}(q-1)b(n-2u-2,r-u-1)\\
&= 1 + \sum_{u=n-r-1}^{r-1}(q-1)q^{n-u-1} + q(q-1)b(n-2,r-1) \\
& \quad + q\sum_{u=0}^{n-r-3} q^{u+1}(q-1)b(n-2u-4,r-u-2)\\
&= 1 + \sum_{u=n-r-1}^{r-1}(q-1)q^{n-u-1} + q(q-1)b(n-2,r-1)\\
& \quad + q\left[ b(n-2,r-1)-1-\sum_{u=n-r-2}^{r-2}(q-1)q^{n-u-3} \right]\\
&= 1-q +q^2b(n-2,r-1)+\sum_{u=n-r-1}^{r-1}(q-1)q^{n-u-1} \\
& \quad -q\sum_{u=n-r-2}^{r-2}(q-1)q^{n-u-3}\\
&=  1-q +q^2b(n-2,r-1) \\
&\quad + (q-1)\left[ \sum_{u=n-r-1}^{r-1}q^{n-u-1} -\sum_{u=n-r-2}^{r-2}q^{n-u-2}  \right]\\
&= 1-q +q^2b(n-2,r-1)  \\
&\quad +(q-1) \left[ \sum_{u=n-r-1}^{r-1}q^{n-u-1} -\sum_{u=n-r-1}^{r-1}q^{n-u-1}  \right]\\ 
&= 1-q +q^2b(n-2,r-1)
\end{align*}
\end{proof}

As a consequence of the corollaries, we get the following theorem.

\begin{thm}\label{thm:7}
Given two integers $r\leq n$, the number of finite sequences of length $n$ with linear complexity at most $r$ is
\begin{enumerate}[(i)]
\item If $r=0$, $b(n,0) = 1$.
\item If $r+1\leq n-r$,
\[
b(n,r) = \frac{q^{2r+1}+1}{q+1}
\]
\item If $n-r\leq r$,
\[
b(n,r) = \frac{1-q^{2(n-r)}}{1+q}+q^n.
\]
\end{enumerate}
\end{thm}
\begin{proof}
Suppose that $r+1\leq n-r$. From the previous corollary,
\[
\sum_{i=0}^{r-1} q^{2i} b(n-2i,r-i) = (1-q)\sum_{i=0}^{r-1} q^{2i} + \sum_{i=0}^{r-1} q^{2(i+1)}b(n-2(i+1),r-(i+1)).
\]
Thus
\[
\sum_{i=0}^{r-1} q^{2i} b(n-2i,r-i) = (1-q)\sum_{i=0}^{r-1} q^{2i} +\sum_{i=1}^{r}q^{2i} b(n-2i,r-i)
\]
Therefore,
\[
b(n,r) = (1-q)\sum_{i=0}^{r-1} q^{2i}+q^{2r} = q^{2r}+\frac{1-q^{2r}}{1+q}
\]
And we get the result.

If $n-r\leq r$, then $r\geq \frac{n}{2}$. So using this,

\[
\sum_{i=0}^{n-r-1} q^{2i} b(n-2i,r-i) = (1-q)\sum_{i=0}^{n-r-1} q^{2i} + \sum_{i=0}^{n-r-1} q^{2(i+1)}b(n-2(i+1),r-(i+1))
\]
Therefore
\[
b(n,r) =(1-q)\sum_{i=0}^{n-r-1} q^{2i} + q^{2\left(n-r\right)}b(2r-n,2r-n),
\]
Since $B(2r-n,2r-n) = \F^{2r-n}$, then
\[
b(n,r) =(1-q)\frac{1-q^{2(n-r)}}{1-q^2} + q^n
\]
And we get our result.

\end{proof}

Using the previous theorem, we can compute the number of finite sequence with a fixed linear complexity.

\begin{thm}\label{thm:8}
Let $r\leq n$ be positive integers. Then, the number of sequences of length $n$ and linear complexity $r$ over a finite field $\F$ of size $q$ is 
\[
\begin{cases}
1\quad &\text{if } r=0,\\
q^{2r-1}(q-1) \quad &\text{if } r\leq \lfloor\frac{n}{2}\rfloor ,\\
q^{2(n-r)}(q-1)\quad &\text{if } r>  \lfloor\frac{n}{2}\rfloor.
\end{cases}
\]
\end{thm}
\begin{proof}
The case $r=0$ is clear. For $r=1$, we get that the number of sequences of length $n$ and linear complexity $r$ over a finite field $\F$ of size $q$ is
\[
\frac{q^3-q}{q+1} = q(q-1).
\]

Now, suppose that $r=\lceil\frac{n}{2}\rceil$. Then the number we want is given by
\[
q^n-\frac{q^{2(n-r)}+q^{2r-1}}{q+1} =\begin{cases}
q^{2r-1}(q-1) & \text{if } n \text{ is even}\\
q^{2(n-r)}(q-1) & \text{if } n \text{ is odd}
\end{cases}
\]
It is easy to check that if $r\leq \lceil\frac{n}{2}\rceil -1$, then the number is
\[
q^{2r-1}(q-1),
\]
and if $\lceil\frac{n}{2}\rceil +1$, the number is
\[
q^{2(n-r)}(q-1).
\]
Furthermore $\lbrace r\leq \lceil\frac{n}{2}\rceil -1\rbrace $ and $\lbrace r=\lceil\frac{n}{2}\rceil, n \text{ even} \rbrace$ are the same as $\lbrace r\leq \lfloor\frac{n}{2}\rfloor$. Finally $\lbrace r\geq \lceil\frac{n}{2}\rceil +1\rbrace $ and $\lbrace r=\lceil\frac{n}{2}\rceil, n \text{ odd} \rbrace$ are the same as $\lbrace r> \lfloor\frac{n}{2}\rfloor$.
\end{proof}

Since we also know the size of balls with respect to the linear complexity from Theorem \ref{thm:7}, we can give a formula for the Sphere packing bound.

\begin{thm}[Sphere packing bound]
Let $S$ be a set of sequences of length $n$ and with minimum distance $d$. Then
\[
\sharp S\leq 
\begin{cases}
\frac{q^n(q+1)}{q^{2\lfloor \frac{d-1}{2} \rfloor}+1} & \text{if } 2\lfloor \frac{d-1}{2} \rfloor\leq n-1, \\
\frac{q^n(q+1)}{1-q^{2(n-\lfloor \frac{d-1}{2} \rfloor)}+(1+q)q^n} & \text{if } 2\lfloor \frac{d-1}{2} \rfloor> n-1. \\
\end{cases}
\]
\end{thm}
\begin{proof}
This is a direct consequence of Theorem \ref{thm:7} and using the fact that the union of the spheres of radius $\lfloor \frac{d-1}{2}\rfloor$ centered at the sequences in $S$ is a disjoint union.
\end{proof}

\section{Conclusion and future work}\label{sec:5}
We have seen how the notion of weight of vectors can be extended to the notion of linear complexity of finite sequences. Using the new metric defined by the linear complexity, we developed a coding theory for finite sequences. We gave the Singleton bound and we presented a construction for an optimal set of sequences reaching this bound. Then we computed an exact formula for the number of finite sequences which can be generated by a LFSR of a fixed order. 

LFSR have been extensively studied \cite{Rue86}. It is widely used in the generation of random secret key in symmetric cryptography. Our main goal however is to use the LFSR and linear complexity to get a new protocols for asymmetric public key cryptography. 

In 1978, McEliece proposed a new cryptosystem using linear codes (Goppa codes) and Hamming metric \cite{McE78}. After 40 years of cryptanalysis, the cryptosystem is still considered to be generally secure. However, the cryptosystem requires the use of public keys with large size. This makes it impractical for daily use. The advantage of using linear codes is that cryptosystem based on them are in general safe against the quantum computers. Namely, there is no general algorithm which can decode a random linear code in polynomial time.

The strength of the McEliece cryptosystem is that the Goppa codes look like random linear codes and it is considered to be a difficult problem to decode a random linear code. 
To solve the problem with the key size, it was suggested to use different family of linear codes. For instance, Niederreiter proposed a new cryptosystem using Reed-Solomon codes \cite{Nie86}. However, cryptosystems using Reed-Solomon codes were proven to be insecure \cite{SS92}. Several types of codes were suggested to get a secure cryptosystem. Another suggestion was that, instead of using the classical Hamming metric on the linear code, one use the rank metric. For instance, a new cryptosystem based on the Gabidulin codes were proposed \cite{GPT91}. This system was still proven to be insecure \cite{Ove08}. 

Recently, this increased the interest in the search of linear codes with good properties which can be used in cryptography both in Hamming and rank metric. There is another cryptosystem which are also using a set and a metric on the set. The lattice based cryptosystem  is the scheme where the metric is the Euclidean distance \cite{Ajt96}. This particular cryptosystem is also resistant against attacks from quantum computers.

Motivated by all of this, we may think of a cryptosystem using the linear complexity as metric. We are working in this direction using the metric from linear complexity and this will be part of a future publication. Finally we all know what Hamming metric codes are good for error correcting in a $q$-ary symmetric channel. For rank metric codes, they have good application in network coding \cite{KK08,SKK08}. It is our hope that the presented coding framework will also be of use for some particular channel.
 
\section*{Aknowledgement}

I would like to thank Prof. Joachim Rosenthal for his valuable comments and suggestions on this work.

\FloatBarrier
\bibliography{reference}

\begin{thebibliography}{{McE}78}

\bibitem[Ajt96]{Ajt96}
M.~Ajtai.
\newblock Generating hard instances of lattice problems (extended abstract).
\newblock In {\em In Proceedings of the Twenty-Eighth Annual ACM Symposium on
  the Theory of Computing}, pages 99--108. ACM, 1996.

\bibitem[Day60]{Day60}
D.~E. Daykin.
\newblock Distribution of bordered persymmetric matrices in a finite field.
\newblock {\em Journal f\"ur die reine und angewandte Mathematik}, 203:47--54,
  1960.

\bibitem[GPT91]{GPT91}
E.~M. Gabidulin, A.~V. Paramonov, and O.~V. Tretjakov.
\newblock {\em Ideals over a Non-Commutative Ring and their Application in
  Cryptology}, pages 482--489.
\newblock Springer Berlin Heidelberg, Berlin, Heidelberg, 1991.

\bibitem[KK08]{KK08}
R.~Koetter and F.R. Kschischang.
\newblock Coding for errors and erasures in random network coding.
\newblock {\em IEEE Transactions on Information Theory}, 54(8):3579--3591, Aug
  2008.

\bibitem[LN96]{Lid96}
R.~{Lidl} and H.~{Niederreiter}.
\newblock {\em {Finite fields. 2nd ed.}}
\newblock Cambridge: Cambridge Univ. Press, 2nd ed. edition, 1996.

\bibitem[{McE}78]{McE78}
R.~J. {McEliece}.
\newblock {A Public-Key Cryptosystem Based On Algebraic Coding Theory}.
\newblock {\em Deep Space Network Progress Report}, 44:114--116, January 1978.

\bibitem[Nie86]{Nie86}
H~Niederreiter.
\newblock Knapsack type cryptosystems and algebraic coding theory.
\newblock {\em Problems of Control and Information Theory. Problemy Upravlenija
  i Teorii Informacii}, 25:19--34, 1986.

\bibitem[Ove08]{Ove08}
R.~Overbeck.
\newblock Structural attacks for public key cryptosystems based on gabidulin
  codes".
\newblock {\em Journal of Cryptology}, 21(2):280--301, Apr 2008.

\bibitem[Rue86]{Rue86}
R.~Rueppel.
\newblock {\em Analysis and Design of Stream Ciphers}.
\newblock Springer Berlin Heidelberg, Berlin, Heidelberg, 1986.

\bibitem[SKK08]{SKK08}
D.~Silva, F.R. Kschischang, and R.~Koetter.
\newblock A rank-metric approach to error control in random network coding.
\newblock {\em IEEE Transactions on Information Theory}, 54(9):3951--3967, Sept
  2008.

\bibitem[SS92]{SS92}
M.~Sidelnikov and O.~Shestakov.
\newblock On insecurity of cryptosystems based on generalized reed-solomon
  codes.
\newblock {\em Discrete Mathematics and Applications}, 2(4):439--444, 1992.

\end{thebibliography}

\end{document}